\documentclass[letterpaper, 10 pt, conference]{ieeeconf}  % Comment this line out if you need a4paper

\IEEEoverridecommandlockouts                              % This command is only needed if 
\usepackage{amsmath,amssymb,amsfonts}
\usepackage[ruled]{algorithm2e}
\usepackage{graphicx}
\usepackage{subfig}
\usepackage{textcomp}
\usepackage{booktabs}
\usepackage{textcomp}

\title{\LARGE \bf
Modelling Resource Contention in Multi-Robot Task Allocation Problems with Uncertain Timing
}

\author{Andrew W. Palmer, Andrew J. Hill, and Steven J. Scheding$^{1}$% <-this % stops a space
\thanks{This work was supported by the Rio Tinto Centre for Mine Automation and the Australian Centre for Field Robotics, University of Sydney, Australia.}% <-this % stops a space
\thanks{$^{1}$ The authors were with the Australian Centre for Field Robotics, University of Sydney, Australia,
        email: andrew.palmer@siemens.com, a.hill@acfr.usyd.edu.au, steve.scheding@team.telstra.com.}%
}

\begin{document}
\bstctlcite{IEEEexample:BSTcontrol}

\begin{table*}
\copyright 2018 IEEE. Personal use of this material is permitted. Permission from IEEE must be obtained for all other uses, in any current or future media, including reprinting/republishing this material for advertising or promotional purposes, creating new collective works, for resale or redistribution to servers or lists, or reuse of any copyrighted component of this work in other works.

The published version of this article can be found at https://doi.org/10.1109/ICRA.2018.8460981
\end{table*}

\pagebreak

\maketitle
\thispagestyle{empty}
\pagestyle{empty}

%%%%%%%%%%%%%%%%%%%%%%%%%%%%%%%%%%%%%%%%%%%%%%%%%%%%%%%%%%%%%%%%%%%%%%%%%%%%%%%%
\begin{abstract}
This paper proposes an analytical framework for modelling resource contention in multi-robot systems, where the travel times and task durations are uncertain. It uses several approximation methods to quickly and accurately calculate the probability distributions describing the times at which the tasks start and finish. Specific contributions include exact and fast approximation methods for calculating the probability of a set of independent normally distributed random events occurring in a given order, a method for calculating the most likely and $n$-th most likely orders of occurrence for a set of independent normally distributed random events that have equal standard deviations, and a method for approximating the conditional probability distributions of the events given a specific order of the events. The complete framework is shown to be faster than a Monte Carlo approach for the same accuracy in two multi-robot task allocation problems. In addition, the importance of incorporating uncertainty is demonstrated through a comparison with a deterministic method. This is a general framework that is agnostic to the optimisation method and objective function used, and is applicable to a wide range of problems.
\end{abstract}

%%%%%%%%%%%%%%%%%%%%%%%%%%%%%%%%%%%%%%%%%%%%%%%%%%%%%%%%%%%%%%%%%%%%%%%%%%%%%%%%
\section{Introduction}
Multi-Robot Task Allocation (MRTA) problems arise in many scenarios, and involve assigning a set of robots to a set of tasks. Much of the existing solution methods assume deterministic robot dynamics, and ignore the effects of resource contention. This paper explicitly models uncertainty in the travel times and task durations of the robots, and the effects that result from multiple robots using mutually exclusive resources. An analytical framework for calculating the cost of a set of task assignments is developed, and it is shown to outperform deterministic and Monte Carlo approaches.

This work is motivated by multi-robot scenarios where there is a shared resource that cannot be used by all of the robots at once, such as an intersection or a recharging point, as well as situations where the robots cannot perform their tasks in parallel and must wait for previous robots to finish before commencing their task, such as construction and maintenance tasks. The techniques developed in this paper for modelling these effects are applicable to a wide range of robotics scenarios such as multi-robot path planning \cite{Alonso-Mora2015} and planning for recharging robots \cite{Palmer2016b}, as well as other scenarios such as machine shop scheduling \cite{Kouvelis2000}. While there is literature on MRTA problems that incorporate either uncertainty \cite{Liu2011, Wu2015, Nam2015, Nam2015a} or resource contention \cite{Vaughan2000, Nam2015b}, to the best of the authors' knowledge this is the first work to combine the two. The framework developed in this paper analytically calculates the probability distributions describing the times at which the tasks are started and completed. This framework is independent of the choice of optimisation method used, facilitates the use of any objective function, and can also be used in conjunction with chance constraints. In addition to the framework, specific contributions of this paper include:

\begin{itemize}
	\item exact and fast approximation methods for calculating the probability of a set of independent normally distributed random events occurring in a given order;
	\item a method for calculating the most likely and $n$-th most likely orderings of independent normally distributed random events, when the standard deviations of the occurrence times for each event are equal; and
	\item a normal approximation to the conditional probability distribution describing a random event given a specific order of events.
%	\item a comparison of the framework with deterministic and Monte Carlo approaches in two simulation examples.
\end{itemize}

In the following sections, Section \ref{s:rel_lit} presents an overview of related literature, and Section \ref{s:framework} develops the analytical framework. Section \ref{s:results} then evaluates the utility of the framework in two simulation examples, and Section \ref{s:conc} concludes the paper with suggestions for future research.

%In the problems considered in this paper, the locations where the robots perform their tasks are the contended resource. Two specific problems are considered. In the first problem, the order in which the robots must perform their tasks is specified in the problem formulation. This behaviour can be found in construction and maintenance problems where specialised robots must perform their tasks first to enable other robots to execute their tasks. For example, one robot may need to remove a wheel to allow another robot to access the brakes on a vehicle. In the second problem, the order is unspecified and the robots perform their tasks in the order in which they arrive at the location. These situations can occur when there are robots or agents that are outside of the control of the system and rules are used to govern their behaviour. Examples of this include fleets of autonomous vehicles interacting with autonomous and manned vehicles at intersections through give way rules, and robots forming queues to perform tasks such as collecting and delivering packages.  

\section{Related literature and background}\label{s:rel_lit}

MRTA problems have been extensively studied in the literature---a recent review of the state-of-the-art solution methods for MRTA problems is presented in \cite{Khamis2015}. The authors note that solving MRTA problems with complex constraints, including uncertainty and resource contention, is still an open question. Uncertainty has been considered by several papers, but has been dealt with by each in different ways. The interval Hungarian algorithm was developed in \cite{Liu2011} to deal with problems that have uncertainty in the utility estimate of a given assignment. This method relies on knowing the Probability Density Function (PDF) describing the utility. The interval Hungarian algorithm can be applied to problems with resource contention, using the framework developed in this paper to calculate the PDFs of the utility. Task allocation in teams consisting of both robots and humans was considered by \cite{Wu2015}, where humans have the option of rejecting a task assignment. They developed a replanning algorithm using a multi-agent Markov decision process that incorporated the probability that a human will reject the task. Finally, sensitivity analysis approaches were used in \cite{Nam2015,Nam2015a} to quantify when a task assignment should be recomputed in response to changes in the environment. The effect of resource contention was not included in any of the above papers. 

Resource contention was considered in \cite{Vaughan2000}, where a team of robots operating in an office building frequently encountered areas where only one robot could operate at a time, such as doorways and cluttered corridors. The authors introduced a decentralised method that used aggression signalling to resolve interactions during task execution. A method for calculating the optimal task assignment in scenarios with resource contention was developed in \cite{Nam2015b}. The costs associated with the resource contention were modelled using a penalisation function, and they showed that the problem is NP-hard when the penalisation function is polynomial-time computable. Their approach used Murty's ranking algorithm to find next best assignments when ignoring contention costs, and then evaluated these assignments with the contention costs included. The above approach could be used in conjunction with the framework developed in this paper for scenarios where uncertainty is also considered. 

The proposed framework developed in this paper uses a number of approximations for performing operations on normal distributions. Two existing approximations that are used in the framework are introduced in the following subsections. 

\subsection{Maximum of normally distributed random variables} \label{s:max}

This subsection presents an approximation for the mean and variance of the maximum of two independent normally distributed variables, originally presented in \cite{Clark1961}. Consider $X \sim \mathcal{N}(\mu_{X}, \sigma_{X}^{2})$ and $Y \sim \mathcal{N}(\mu_{Y}, \sigma_{Y}^{2})$. Let
\begin{equation}
\alpha = \sqrt{\sigma_{X}^{2} + \sigma_{Y}^{2}}, \qquad \beta = \frac{\mu_{X} - \mu_{Y}}{\alpha}.
\end{equation}
Using the following notation:
\begin{equation} \label{eq:normal_pdf}
\phi(x) = \frac{\exp(-x^{2}/2)}{\sqrt{(2\pi)}},
\end{equation}
\begin{equation} \label{eq:normal_cdf}
\Phi(x) = \int_{-\infty}^{x}\phi(t) \textrm{d}t = \frac{1}{2}\left(1 + \textrm{erf}\left(\frac{x}{\sqrt{2}}\right)\right),
\end{equation}
where erf(.) is the error function, defined as
\begin{equation}
\textrm{erf}(t) = \frac{2}{\sqrt{\pi}}\int\limits_{0}^{t}\exp\left(-\tau^{2}\right) \textrm{d}\tau,
\end{equation}
$\max(X,Y)$ is approximated by a normal distribution, $Z \sim \mathcal{N}(\mu_{Z}, \sigma_{Z}^{2})$, where
\begin{equation}
\mu_{Z} = \mu_{X}\Phi(\beta) + \mu_{Y}\Phi(-\beta) + \alpha\phi(\beta),
\end{equation}
\begin{multline}
\sigma_{Z}^{2} = (\mu_{X}^{2} + \sigma_{X}^{2})\Phi(\beta) + (\mu_{Y}^{2} + \sigma_{Y}^{2})\Phi(-\beta)\\
+ (\mu_{X} + \mu_{Y})\alpha\phi(\beta)   -\mu_{Z}^{2}. 
\end{multline}
For more than two variables, the author suggests recursively applying the above approximation to pairs of variables. 

%The equations above have been simplified from the full method presented in \cite{Clark1961}. Where the above equations are only valid for independent random variables, the full approach is valid for correlated random variables. While situations where the random variables are correlated are not considered in this paper, correlation may need to be accounted for in MRTA scenarios where the planner generates a sequence of tasks for each robot rather than just a single task.

\subsection{Conditioning normally distributed random variables} \label{s:cond}

This subsection summarises a method presented in \cite{Palmer2016} for calculating the mean and variance of a normally distributed random variable, $B$, that is conditioned on other independent normally distributed random variables, $A$ and $C$, to satisfy the inequality $A < B < C$. This method can be used to calculate $(B|B<C)$ and $(B|A<B)$ by using $\mu_{A} = -\infty$ and $\mu_{C} = \infty$ respectively. The mean and variance of the conditional probability distribution are denoted as $\hat{\mu}_{B}$ and $\hat{\sigma}_{B}^{2}$ respectively. First, the random variables are transformed such that $B$ is described by a standard normal distribution. This yields transformations of $A$ and $C$ to $D$ and $E$ respectively:
\begin{equation} \label{eq:begin_conditional}
D \sim \mathcal{N}(\mu_{D},\sigma_{D}^{2}), \qquad E \sim \mathcal{N}(\mu_{E},\sigma_{E}^{2}), 
\end{equation}
where:
\begin{equation}\label{eq:transform1}
\mu_{D} = \frac{\mu_{A} - \mu_{B}}{\sigma_{B}}, \qquad \sigma_{D}^{2} = \frac{\sigma_{A}^{2}}{\sigma_{B}^{2}},
\end{equation}
\begin{equation}\label{eq:transform2}
\mu_{E} = \frac{\mu_{C} - \mu_{B}}{\sigma_{B}}, \qquad \sigma_{E}^{2} = \frac{\sigma_{C}^{2}}{\sigma_{B}^{2}}.
\end{equation}
Then, the mean of the conditional probability distribution is given by
\begin{multline}
\mu_{N} = 2\alpha\left(\frac{1}{\sqrt{\sigma_{D}^2 + 1}}\exp\left(-\frac{\mu_{D}^{2}}{2(\sigma_{D}^2 + 1)}\right) \right. \\
\qquad \left. - \frac{1}{\sqrt{\sigma_{E}^2 + 1}}\exp\left(-\frac{\mu_{E}^{2}}{2(\sigma_{E}^2 + 1)}\right)\right),
\end{multline}
where
\begin{equation}
\alpha = \frac{1}{\sqrt{2\pi}\left[\textrm{erf}\left(\frac{\mu_{E}}{\sqrt{2\left(\sigma_{E}^{2} + 1\right)}}\right) - \textrm{erf}\left(\frac{\mu_{D}}{\sqrt{2\left(\sigma_{D}^{2} + 1\right)}}\right)\right]},
\end{equation}
and the variance is given by
\begin{equation}
\begin{aligned}
\sigma_{N}^{2} &= \alpha \left[ \sqrt{2\pi} \left( \left( 1 + \mu_{N}^{2} \right) \left( \textrm{erf} \left( \frac{\mu_{E}}{\sqrt{2(\sigma_{E}^{2} + 1)}} \right) \right. \right. \right. \\
&- \left. \left. \left.  \textrm{erf}\left(\frac{\mu_{D}}{\sqrt{2(\sigma_{D}^{2}+1)}} \right) \right) \right) \right. \\ 
& + \left. \frac{2}{\sqrt{\sigma_{D}^2 + 1}}\left(\frac{\mu_{D}}{\sigma_{D}^2 + 1} - 2\mu_{N} \right)\exp\left(-\frac{\mu_{D}^{2}}{2(\sigma_{D}^2 + 1)}\right) \right. \\
& - \left. \frac{2}{\sqrt{\sigma_{E}^2 + 1}}\left(\frac{\mu_{E}}{\sigma_{E}^2 + 1} - 2\mu_{N} \right)\exp\left({-\frac{\mu_{E}^{2}}{2(\sigma_{E}^2 + 1)}}\right)\right].
\end{aligned}
\end{equation}
The new mean and standard deviation are then transformed back to the original reference frame to give $\hat{\mu}_{B}$ and $\hat{\sigma}_{B}^{2}$:
\begin{equation} \label{eq:end_conditional}
\hat{\mu}_{B} = \mu_{N} \sigma_{B} + \mu_{B}, \qquad \hat{\sigma}_{B}^{2} = \sigma_{N}^{2} \sigma_{B}^{2}. 
\end{equation}
The approach used to derive the above method is reliant on the condition that $\mu_{A} < \mu_{C}$, and the assumption that the probability distributions describing $A$ and $C$ have limited overlap.

\section{Framework} \label{s:framework}

This section develops the analytical framework for modelling resource contention when uncertainty is considered. Two cases are examined---in the first, the order in which the robots must use the resource is specified, while in the second, the order is simply the order in which they arrive. The robots have a time that they arrive at and begin queuing for the resource, $T^{a}_{X}$, and a duration for using the resource, $D_{X}$, that are independent normally distributed random variables, where $X$ identifies the robot. 

For the first case, consider a simple example of two robots, $A$ and $B$, that both need to perform an action at the same location. Only one robot can perform their action at the location at a time, so the second robot may have to wait for the first robot to complete its action before commencing its own action. If robot $A$ must perform its action before $B$, then the time that robot $A$ starts its action, $T^{s}_{A}$, is simply the time that it arrives at the location, $T^{a}_{A}$. The time that it finishes the action at, $T^{f}_{A}$, is given by:
\begin{equation} \label{e:a_finish}
T^{f}_{A} = T^{s}_{A} + D_{A},
\end{equation}
where the mean and variance of $T^{f}_{A}$ are the sum of the means and variances respectively of $T^{s}_{A}$ and $D_{A}$. Since robot $B$ can only commence its action after $A$ has finished, the time that robot $B$ starts its action, $T^{s}_{B}$ is given by:
\begin{equation} \label{eq:max}
T^{s}_{B} = \max(T^{f}_{A}, T^{a}_{B}),
\end{equation}
where the $\max$ is calculated using the approximation described in Section \ref{s:max}. The time that robot $B$ completes its action, $T^{f}_{B}$, is then given by:
\begin{equation}\label{eq:finish}
T^{f}_{B} = T^{s}_{B} + D_{B}.
\end{equation}
For 3 or more robots, the start and finish times of their actions are calculated by iteratively applying Eqs. (\ref{eq:max}) and (\ref{eq:finish}) with respect to the previous robot to perform its task.

The second case, where the order in which the robots use the resource is not specified but is instead determined by when the robot begins queuing for the resource, exhibits a First-In First-Out (FIFO) property. Consider a similar multi-robot scenario to the first case where two robots, $A$ and $B$, travel to a location and perform a task. In contrast to the first case, the first robot to arrive at the location is the first to perform its task. For two robots, there are two possible orders in which the tasks can be performed---$A$ followed by $B$, and $B$ followed by $A$. To calculate the completion time for each robot, both orders must be considered. 
% Note - I didn't change the notation to say, if T^{a}_{A} < T^{a}_{B}, because it will make it confusing for when the weighted completion time is calculated. 
Consider the case where $A$ arrives before $B$. The time that $A$ starts its task is first conditioned on the order in which the robots arrive:
\begin{equation}
(T^{s}_{A}|T^{a}_{A} < T^{a}_{B}) = (T^{a}_{A}|T^{a}_{A} < T^{a}_{B}). 
\end{equation}
The time that $A$ completes its task is given by: 
\begin{equation}
(T^{f}_{A}|T^{a}_{A} < T^{a}_{B}) = (T^{s}_{A}|T^{a}_{A} < T^{a}_{B}) + D_{A}. 
\end{equation}
The time that $B$ starts its task is then calculated as:
\begin{equation}
(T^{s}_{B}|T^{a}_{A} < T^{a}_{B}) = \max \left((T^{f}_{A}|T^{a}_{A} < T^{a}_{B}), (T^{a}_{B}|T^{a}_{A} < T^{a}_{B})   \right),
\end{equation}
and the time that $B$ completes its task is given by:
\begin{equation}
(T^{f}_{B}|T^{a}_{A} < T^{a}_{B}) = (T^{s}_{B}|T^{a}_{A} < T^{a}_{B}) + D_{B}. 
\end{equation}
The completion time for each robot, considering all orders of arrival, is calculated by summing the probability weighted completion times for each order of arrival. For robot $A$:
\begin{multline}
T^{f}_{A} = P(T^{a}_{A} < T^{a}_{B})(T^{f}_{A}|T^{a}_{A} < T^{a}_{B}) \\
+ P(T^{a}_{B} < T^{a}_{A})(T^{f}_{A}|T^{a}_{B} < T^{a}_{A}),
\end{multline}
where $P(T^{a}_{X} < T^{a}_{Y})$ is the probability that $X$ arrives before $Y$. This process extends to $n$ robots, with the downside that there are $n!$ orders of arrival that have to be considered. 

The following subsections present methods for calculating the arrival times conditioned on the order of arrival (Section \ref{s:conditioning}) and the probability of an order of arrival (Section \ref{s:prob}). A method for reducing the number of orders of arrival that are considered based on Mutry's ranking algorithm is developed in Section \ref{s:most_likely}. 
%Section \ref{s:complexity} discusses some possible approaches for reducing the computational requirements of the methods. 

\subsection{Conditioning the arrival times on the order of arrival} \label{s:conditioning}

A method for calculating the mean and variance of the random variable $T^{a}_{B}$ conditioned on the arrival order $ABC$, $(T^{a}_{B} | T^{a}_A < T^{a}_B < T^{a}_C)$, was presented in Section \ref{s:cond}. 

\subsubsection{Extension to general conditions}

For cases where the conditions listed in Section \ref{s:cond} are not met, the method is a poor approximation of the mean and variance. To apply this approach in such situations, it is proposed that $T^{a}_{A}$ and $T^{a}_{C}$ be applied one at a time. Applying these conditions one at a time involves using the procedure from Eq. (\ref{eq:begin_conditional}) to Eq. (\ref{eq:end_conditional}) with one of the distributions set to either $-\infty$ or $\infty$, and then using the result of that procedure as the distribution for $T^{a}_{B}$ when applying the remaining condition. 

The order in which the conditions are applied can impact the resultant distribution. In order to determine the order in which the conditions should be applied, a decision tree was learnt using the scikit-learn \cite{scikit-learn} module for Python. The three choices of method are:
\begin{enumerate}
	\item apply $T^{a}_{A}$ and $T^{a}_{C}$ together;
	\item apply $T^{a}_{A}$ followed by $T^{a}_{C}$; and
	\item apply $T^{a}_{C}$ followed by $T^{a}_{A}$. 
\end{enumerate}

The decision tree to determine which method to use is presented in Figure \ref{f:decision_tree}. The parameters $\gamma$ and $\delta$ are an overlap metric and shape metric respectively, as defined in \cite{Palmer2016}:
\begin{equation}
\gamma = \frac{\mu_{C} - \mu_{A}}{\sigma_{C} + \sigma_{A}}, \qquad \delta = \left| \log\left(\frac{\sigma_{A}}{\sigma_{C}}\right)\right|.
\end{equation}
%$\gamma$ is a measure of the degree of overlap between the probability distributions of the two conditions, and $\delta$ is a measure of the difference in the variances of $T^{a}_{A}$ and $T^{a}_{C}$. 

\begin{figure}
	\centering
	\includegraphics[width = 0.4\textwidth]{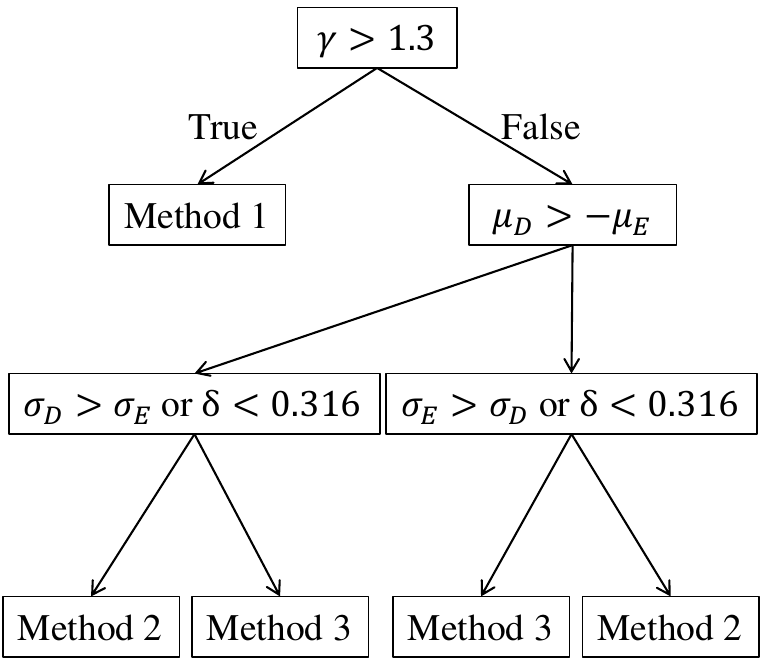}
	\caption{Decision tree for determining the order in which to apply the conditions. Note that the means and standard deviations used here are calculated from the original distributions using Eq.~(\ref{eq:transform1}) and (\ref{eq:transform2}). The first decision checks if the distributions overlap enough to require the conditions to be applied individually, while the second and third decisions compare the relative positions and shapes of the two distributions to determine which should be applied first. } \label{f:decision_tree}
\end{figure}

The decision tree was trained on over 200,000 different combinations of values for $T^{a}_{A}$ and $T^{a}_{C}$, and a separate set of 10,000 combinations of randomly sampled values was used for evaluation. Kullback-Leibler (KL) divergence was used as a measure of error between the method selected by the decision tree and the best method. The best method was selected in 84.8\% of cases. For all cases, the average and Root Mean Squared (RMS) KL divergences were $3.1\text{E}{-3}$ and $9.4\text{E}{-2}$ respectively, and for cases where the incorrect decision was made they were $2.0\text{E}{-2}$ and $2.4\text{E}{-1}$ respectively. The majority of the error was accrued in a few cases where the distributions $T^{a}_{A}$ and $T^{a}_{C}$ have a very low probability of satisfying $T^{a}_{A} < T^{a}_{C}$, resulting in a very high KL divergence ($>1$). Thus, even though this is a poor approximation in a small number of cases, it will have a negligible impact on the result when it is multiplied by the probability of that order occurring.

\subsubsection{Iteratively applying conditions}

So far, only the second robot to arrive in a group of three robots has been considered. Analytically calculating the mean and variance for the conditional arrival time of the other robots is a challenging problem. Instead, it is proposed that the conditions be applied iteratively, thus enabling any number of robots to be considered. For example, the conditional arrival time of the first robot, $(T^{a}_{A} | T^{a}_{A} < T^{a}_{B} < T^{a}_{C})$, can be calculated as:
\begin{equation}
(T^{a}_{A} | T^{a}_{A} < T^{a}_{B} < T^{a}_{C}) = (T^{a}_{A} | T^{a}_{A} < (T^{a}_{B}|T^{a}_{B} < T^{a}_{C})). 
\end{equation}

\subsection{Calculating the probability of an order of arrival} \label{s:prob}

First, consider two robots with arrival times $T^{a}_{A}$ and $T^{a}_{B}$. The probability that $A$ arrives before $B$ is given by:
\begin{equation} \label{eq:prob}
P(T^{a}_{A} < T^{a}_{B}) = \int\limits_{-\infty}^{0}p_{T^{a}_{A} - T^{a}_{B}}(t) \textrm{d}t,
\end{equation}
where $t$ is time, and $p_{X}(t)$ denotes the PDF of the random variable $X$. If $T^{a}_{A} \sim \mathcal{N}(\mu_{A},\sigma_{A}^{2})$ and $T^{a}_{B} \sim \mathcal{N}(\mu_{B},\sigma_{B}^{2})$, then:
\begin{equation} \label{eq:prob_cdf}
P(T^{a}_{A} < T^{a}_{B}) = \frac{1}{2}\left(1 + \textrm{erf} \left(\frac{\mu_{B} - \mu_{A}}{\sqrt{2\left(\sigma_{A}^{2} + \sigma_{B}^{2} \right)}}  \right)  \right). 
\end{equation}

Extending to three robots, $A$, $B$, and $C$, the probability that $A$ arrives before $B$, and that $B$ arrives before $C$, is:
\begin{multline}
P((T^{a}_{A} < T^{a}_{B}) \cap (T^{a}_{B} < T^{a}_{C})) \\
= P(T^{a}_{A} < T^{a}_{B}) \times P(T^{a}_{B} < T^{a}_{C} | T^{a}_{A} < T^{a}_{B}). 
\end{multline}
%Calculating the conditional probability $P(T^{a}_{B} < T^{a}_{C} | T^{a}_{A} < T^{a}_{B})$ is challenging. 
It is possible to formulate this as a multivariate normal distribution through a linear transformation. Let $X = T^{a}_{A} - T^{a}_{B}$ and $Y = T^{a}_{B} - T^{a}_{C}$, then:
\begin{equation}
\pmb{M} = \begin{bmatrix}
X \\
Y
\end{bmatrix} = \pmb{ST} = \begin{bmatrix}
1 & -1 & 0\\
0 & 1 & -1
\end{bmatrix}
\begin{bmatrix}
T^{a}_{A}\\
T^{a}_{B}\\
T^{a}_{C}
\end{bmatrix},
\end{equation}
\begin{equation}
\pmb{\mu} = \begin{bmatrix}
\mu_{A}\\
\mu_{B}\\
\mu_{C}
\end{bmatrix}, \qquad  
\pmb{\Sigma} = \begin{bmatrix}
\sigma_{A}^{2} & 0 & 0\\
0 & \sigma_{B}^{2} & 0\\
0 & 0 & \sigma_{C}^{2}
\end{bmatrix}. 
\end{equation}
The multivariate normal distribution, $\pmb{M}$, then has mean, $\pmb{\mu}_{\pmb{M}}$ and covariance, $\pmb{\Sigma}_{\pmb{M}}$, calculated using a linear transformation:
\begin{equation}
\pmb{\mu}_{\pmb{M}} = \pmb{S}\pmb{\mu}, 
\end{equation}
and:
\begin{equation}
\pmb{\Sigma}_{\pmb{M}} = \pmb{S}\pmb{\Sigma}\pmb{S}^{T}. 
\end{equation}

The probability $P((T^{a}_{A} < T^{a}_{B}) \cap (T^{a}_{B} < T^{a}_{C}))$ is calculated by evaluating the Cumulative Distribution Function (CDF) of $\pmb{M}$ at $X = 0$ and $Y = 0$. Unfortunately, no analytical solution exists for the CDF of a multivariate normal distribution \cite{Genz2009}. However, a numerical approximation approach based on \cite{Genz1992} is readily available as the \texttt{mvnun} function in the \texttt{stats.mvn} module of the Scipy package for Python. 

The \texttt{mvnun} function is computationally expensive, especially for high-dimensional multivariate distributions. If there are independent parts of the distribution (e.g., $P((T^{a}_{A} < T^{a}_B) \cap (T^{a}_B < T^{a}_C) \cap (T^{a}_D < T^{a}_E) \cap (T^{a}_E < T^{a}_F))$ where $ABC$ are independent from $DEF$), then, in practice when using \texttt{mvnun}, it is significantly faster to calculate the probabilities of each independent part separately and simply multiply the probabilities together, than to compute the probability using the entire multivariate distribution. 

\subsubsection{Estimating the probability} \label{s:prob_est}

Due to the large number of possible orders of arrival when considering multiple robots, calculating the probability of every order of arrival using \texttt{mvnun} can result in long computation times. For example, with 8 robots, there are over 40,000 possible orders of arrival. In this case, \texttt{mvnun} takes 3ms to compute the probability of one arrival order, requiring a total of 120s to calculate the probability of every order. It is therefore desirable to have a fast method of estimating the probability to use either in place of \texttt{mvnun} or to allow unlikely arrival orders to be discarded before \texttt{mvnun} is called. Using conditional probabilities, the probability of a given arrival order of $n$ robots is: 
\begin{multline} \label{eq:est_prob}
\bar{P}\left((T^{a}_{1} < T^{a}_{2}) \cap \dots \cap  (T^{a}_{n-1} < T^{a}_{n})  \right) \\ 
= \prod_{i=2}^{i=n} P \left(\left(T^{a}_{i-1} | \left(T^{a}_{j-2} < T^{a}_{j-1} \; \forall j \in \{3,\dots,i\} \right)\right) < T^{a}_{i}\right). 
\end{multline}
To estimate this probability, the conditional probability distributions are approximated as normal distributions using the approach presented in Section \ref{s:conditioning}, and the probability of each pair of distributions is then calculated using Eq.~\eqref{eq:prob_cdf}.

%This method is approximately 100 times faster than \texttt{mvn.mvnun} in the case of 8 robots, with this advantage increasing as the number of robots is increased. 

\subsection{Finding the $n$-th most likely order of arrival} \label{s:most_likely}

Another method of reducing the computational requirements of the framework is to only consider likely orders of arrival. Thus, it is desirable to be able to determine what the most likely orders of arrival are. In the general case, this requires an exhaustive search over all orders of arrival. This section considers the special case where the standard deviations of the arrival time for each robot are equal. In this case, the $n$-th most likely order of arrival can be found using a similar approach to Murty's ranking algorithm for efficiently ranking assignments by their cost \cite{Murty1968}. 

%This section will first show that the most likely order of arrival in this special case is found by sorting the arrival times by their means. Then, it will show that the next most likely order of arrival is at most one pairwise swap away from the most likely order. Finally, the full approach will be detailed. 

%An approach like branch and bound combined with the upper bound provided by the estimate in Section \ref{s:prob_est} is one way to reduce the size of this search space. 
%

\newtheorem{theorem}{Theorem}
\newtheorem{lemma}{Lemma}

\newtheorem{corollary}{Corollary}

\begin{theorem} \label{th:likely_arrival_order}
	The most likely order of arrival is the order in which the mean arrival times are ascending. 
\end{theorem}

\begin{proof}
	See Appendix \ref{ap:theorem_1_proof} for full proof. 
	
	Summary: The proof presented in Appendix \ref{ap:theorem_1_proof} shows that ordering any pair of neighbouring robots in a sequence of arriving robots by their mean arrival times will result in a higher probability order of arrival than the opposite ordering. Applying this to all pairs of robots leads to the conclusion that the most likely order of arrival is obtained by sorting the robots by their mean arrival times. 
\end{proof}

\begin{corollary} \label{th:pairwise_swap}
	The second most likely order of arrival can be found by swapping one pair of neighbouring robots in the most likely order of arrival. 
\end{corollary}

%\begin{proof}
%%See Appendix \ref{ap:theorem_2_proof}. 	
%The proof for Theorem \ref{th:likely_arrival_order} showed that ordering any pair of neighbouring robots in a sequence of robots by their mean arrival times will result in a higher probability order of arrival than the opposite ordering. Thus, the second most likely order of arrival must be one swap away from the most likely. 
%\end{proof}

\begin{corollary} \label{c:1}
	The $n$-th most likely order of arrival can be found by swapping at most $n-1$ pairs of neighbouring robots when starting with the most likely order.
\end{corollary}

\begin{corollary} \label{c:2}
	The $(n+1)$-th most likely order of arrival can be found by swapping one pair of neighbouring robots in one of the $n$ most likely orders. 
\end{corollary}

Together, Corollaries \ref{c:1} and \ref{c:2} enable the application of an approach similar to Murty's ranking algorithm, detailed in Algorithm \ref{a:nth_most_likely}. The algorithm takes as input the list of arrival time distributions and a probability threshold, and returns a list of orders of arrival that has a sum probability greater than the input threshold. It first finds the most likely order of arrival on Line \ref{l:most_likely}. While the sum probability is less than the threshold (Line \ref{l:threshold}), neighbouring orders of the most likely orders of arrival found up to that point are considered. The $i$-th most likely order is then found on Line \ref{l:ith_most_likely}. 
%Finding a list of orders up to a sum probability threshold is more descriptive than finding the $n$ most likely orders as there are many cases where the $n$ most likely orders are low in probability and the set is unrepresentative of the overall behaviour. 
While this approach is only guaranteed to search through the orders of arrival in the order of their likelihood for cases where the distributions describing the arrival times have equal standard deviations, it can be applied to cases where the standard deviations are not equal with the loss of this guarantee.

%Algorithm \ref{a:nth_most_likely} details an approach which finds the set of most likely orders of arrival up to a threshold on the sum of their probabilities. While the sum of the probabilities of the most likely orders is less than 
%
%It then finds the $i$-th most likely order on Line \ref{l:ith_most_likely} by searching the neighbours of all orders up to the $(i-1)$-th most likely order, disregarding orders that have already been considered. This is necessary as an order can be reached in many ways from different parent nodes. When the sum of the probabilities of the orders considered is greater than the threshold, the algorithm returns the list of orders. 

\begin{algorithm}[t]
	\DontPrintSemicolon
	\SetAlgoNoEnd
	\SetKwFunction{CalcProb}{CalcProb}\SetKwFunction{GetMostLikely}{GetMostLikely}
	\SetKwInOut{Input}{Input}\SetKwInOut{Output}{Output}
	%\GetMostLikely{$\boldsymbol{T}$, $\phi$}
	
	\Input{List of arrival time distributions, $\boldsymbol{T}$, probability threshold, $\phi$}
	\Output{List of orders of arrival, $\boldsymbol{L}$}
	
	\nl $\boldsymbol{O} \leftarrow \boldsymbol{T}$ sorted by mean values \; \label{l:most_likely}
	\nl $\boldsymbol{L} \leftarrow \textrm{list}[\boldsymbol{O}]$, $i \leftarrow 2$ \;
	\nl \While{$\sum \left(p\left(\boldsymbol{Q}\right) \; \forall \boldsymbol{Q} \in \boldsymbol{L}\right) < \phi$}{ \label{l:threshold}
		\nl \For{$\boldsymbol{Y} \in $ unvisited neighbours of $\boldsymbol{O}$\label{l:neighbours}}{
			\nl Append $\boldsymbol{Y}$ to $\boldsymbol{L}$ \;
		}
		\nl Sort $\boldsymbol{L}$ by probability \;
		\nl $\boldsymbol{O} \leftarrow i$-th entry in $\boldsymbol{L}$ \; \label{l:ith_most_likely}
		\nl $i \leftarrow i+1$ \;
	}
	\caption{Calculate the most likely orders of arrival} \label{a:nth_most_likely}
\end{algorithm} 

%\subsection{Computational complexity} \label{s:complexity}
%
%As mentioned previously, there are $n!$ possible orders that $n$ robots can arrive in. This means that, as the number of robots is increased, the analytical approach presented above can quickly become computationally more expensive than numerical methods. The method of estimating the probability presented in Section \ref{s:prob_est} reduces the computational cost of each order considered, and the threshold method developed in Section \ref{s:most_likely} may reduce the number of orders considered. Another simple method of reducing the computational complexity is to examine pairs of robots and disregard a certain arrival order for that pair of robots if the probability is low enough. Consider the robots $A$, $B$, and $C$. If the probability that $B$ arrives before $A$, evaluated using Eq. (\ref{eq:prob_cdf}), is lower than a user specified threshold, then, of the six possible orders that the robots can arrive in, only $ABC$, $ACB$, and $CAB$ would be evaluated. 

\section{Results} \label{s:results}

This section presents results for two MRTA scenarios. In the first scenario, the order in which the robots must perform their tasks was specified, and in the second scenario, the order that the tasks are performed in was determined by the order in which the robots arrive at the task. In both cases, it was assumed that the robots have local collision avoidance algorithms to prevent collisions between robots, and any deviations in their travel times due to collision avoidance between the robots was assumed to be captured by the uncertainty of their travel times. All methods were tested on the same hardware and were programmed in Python. 

\subsection{Scenario 1: Task order is specified}

The first scenario consisted of a heterogeneous fleet of 4 different types of robots performing assembly tasks at four locations. There were 4 robots of each type, giving a total of 16 robots. As the tasks must be executed in a specific order, some robots may have to wait for others to complete their tasks before commencing its own task. The problem was to allocate one robot of each type to each location. 100 random instances were considered. In each instance, the means and standard deviations of the distributions describing the uncertain travel times and uncertain task durations of each robot were randomly selected. The optimisation objective was to minimise the expected cost of the construction, where cost is incurred if the construction takes longer than a specified deadline, $\kappa$. %The deadline used for each scenario was the minimum time to complete the construction if the robots had no uncertainty, as calculated using a deterministic approach. 
The analytical framework presented in this paper enables the calculation of the random variable describing the time at which the construction is completed, $T$. From this, the expected tardiness, $\pi$, is calculated as:
\begin{align}
\pi &= E(\max(0, T-\kappa)) \notag \\
& = \frac{\mu_{T} - \kappa}{2} \left(1 + \textrm{erf} \left( \frac{\mu_{T} - \kappa}{\sigma_{T} \sqrt{2}} \right) \right) \notag \\
& \qquad + \frac{\sigma_{T}}{\sqrt{2 \pi}} \exp \left( -\frac{(\mu_{T} - \kappa)^{2}}{2 \sigma_{T}^{2}} \right),
\end{align}
where $E(X)$ is the expected value of the $X$. 

Deterministic (D) and Monte Carlo (M) approaches were used as benchmarks for the analytical method (A), with the number of samples used in the Monte Carlo method varied between 5 and 80. An exhaustive search over all possible allocations was used for testing each method, and ground truth costs were evaluated using a Monte Carlo method with 100,000 samples. Figure \ref{f:scen1results} shows the performance of each method versus its calculation time. For each random instance, the cost of the lowest cost assignment found by any method was subtracted from the cost for each method and then averaged across all instances. 
%It was necessary to use this relative cost to enable comparison between scenarios. 
As can be seen, the analytical approach presented in this paper consistently resulted in the lowest cost allocations. Not considering uncertainty, as demonstrated by the D result, lead to additional costs that were several orders of magnitude higher than the proposed A approach, highlighting the importance of considering uncertainty in these problems. 

\begin{figure}
	\centering
	\includegraphics[width=0.45\textwidth]{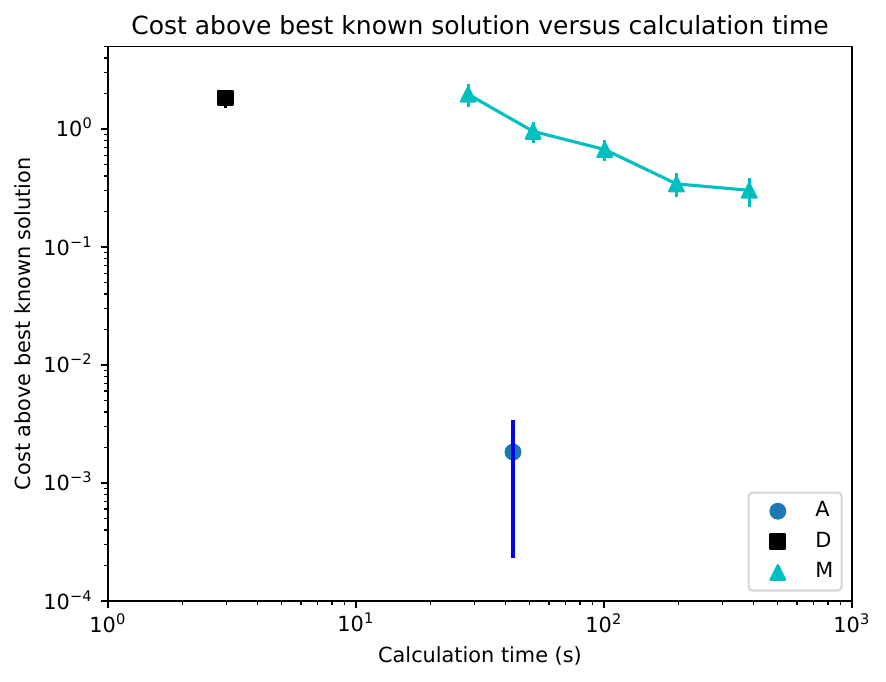}
	
	\caption{Results for Scenario 1 showing the average of the cost minus the best known cost for each scenario versus the calculation time to find the optimal allocation for each method. The error bars show a 95\% confidence interval. The Monte Carlo method used 5, 10, 20, 40, and 80 samples. }
	\label{f:scen1results}
\end{figure}

\subsection{Scenario 2: Task order is unspecified}

The second scenario consisted of 30 controlled robots collecting packages from separate collection locations and delivering them to their destinations. A set of uncontrolled delivery robots were also operating in the environment, collecting packages from the same collection locations. The resources under contention in this scenario are the package collection locations, as the controlled robots may be required to queue at the collection locations before collecting their packages. It was assumed that the actions of the uncontrolled robots were known to the optimiser. The optimisation aim was to generate an optimal allocation of robots to packages that minimised the expected tardiness cost incurred for delivering packages after their deadline. In contrast to the previous scenario, individual deadlines were considered for each package. The number of robots controlled by other entities that visit each collection location, $n$, was varied between 1 and 6. When $n=6$, the entire scenario consisted of 30 controlled robots and 180 uncontrolled robots. 100 random scenarios were tested for each value of $n$, with random arrival time distributions and task durations for each controlled and uncontrolled robot, and random deadlines for each package. The standard deviations of the arrival time distributions for each robot were set to be equal to evaluate the utility of the method presented in Section \ref{s:most_likely} for calculating the set of most likely orders of arrival. 
%Note that the results presented here are valid for cases where the standard deviations are not equal, with the exception of the methods using Algorithm \ref{a:nth_most_likely} with a threshold less than 100\%. 

The analytical method was tested using both \texttt{mvnun} to calculate the probability of an order of arrival (A), and the estimation presented in Section \ref{s:prob_est} (AEst). Algorithm~\ref{a:nth_most_likely} was also used in conjunction with A to limit the number of orders considered, and values of 80\%, 90\%, 95\%, 99\% and 100\% were used for the probability threshold. These approaches were benchmarked against a deterministic approach (D), and a Monte Carlo approach (M). The number of samples used in M was varied between 10 and 10,000 for the cases where $n \le 5$, and between 10 and 100,000 for the case where $n=6$. As only one controlled robot goes to each collection location, the cost incurred by one of the controlled robots does not depend on the assignments of the other controlled robots. Therefore, the cost associated with each robot being assigned to each task can be calculated prior to the optimisation and the Hungarian algorithm \cite{Kuhn1955} can then be used to calculate the optimal allocation. Ground truth costs were evaluated using a Monte Carlo method with 100,000 samples when $n \le 5$, and 1,000,000 samples when $n=6$. 

Figure \ref{f:scen2_results} shows the results for each method versus the calculation time for $n\in\{1,3,6\}$. As can be seen, the A method consistently achieved the lowest cost solutions when using a probability threshold of 100\%. This came at the expense of calculation time as the number of external robots was increased. Lowering the probability threshold resulted in lower calculation times at the expense of solution cost. This highlights the necessity of considering as many of the orders of arrival as possible in order to calculate an accurate representation of the cost of an assignment. 
%The danger of using a probability threshold of less than 100\% and ignoring some of the possible task orders is highlighted by the cases where increasing the probability threshold actually resulted in worse performance. This behaviour can be seen in Figures \ref{sf:scen2_2}, \ref{sf:scen2_3}, and \ref{sf:scen2_6}. 
The AEst approach gave a reduction in calculation time of a factor of approximately 5 over the A approach with a probability threshold of 100\%, with only a marginal increase in the solution cost. In the case where $n=6$, the proposed approaches produced similar results to the M approach using 100,000 samples. However, the proposed approaches had the benefit of higher consistency in this case. Similar to the previous scenario, the D results were several orders of magnitude higher than the approaches incorporating uncertainty.

\begin{figure}
	\centering
	\subfloat[1 uncontrolled robot per location ($n=1$)]{
		\includegraphics[width=0.44\textwidth]{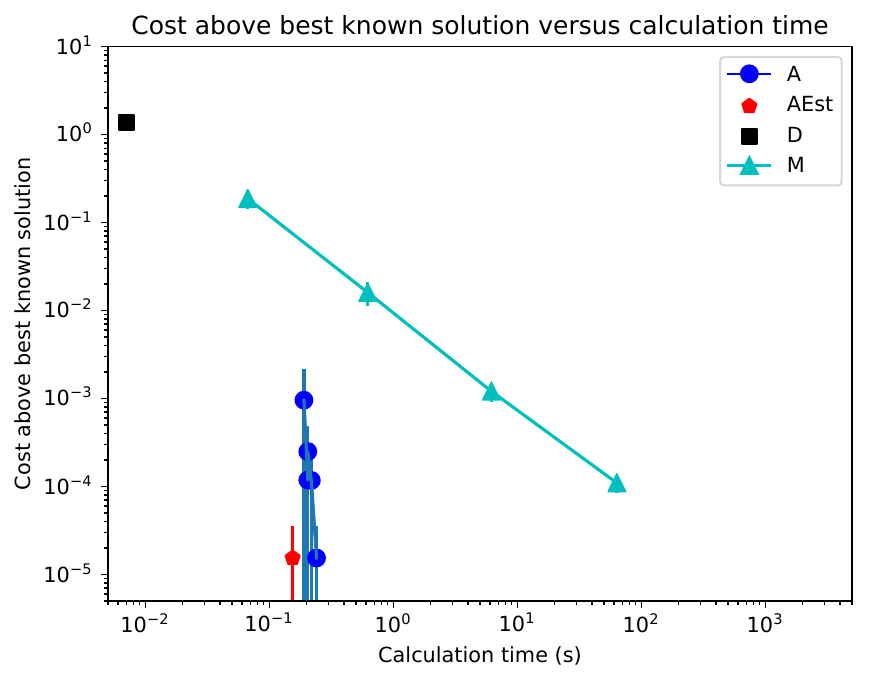}\label{sf:scen2_1}
	}	
	%	\subfloat[$n=2$]{
	%		\includegraphics[width=0.4\textwidth]{scen2results_2}\label{sf:scen2_2}
	%	}
	
	\subfloat[3 uncontrolled robots per location ($n=3$)]{
		\includegraphics[width=0.44\textwidth]{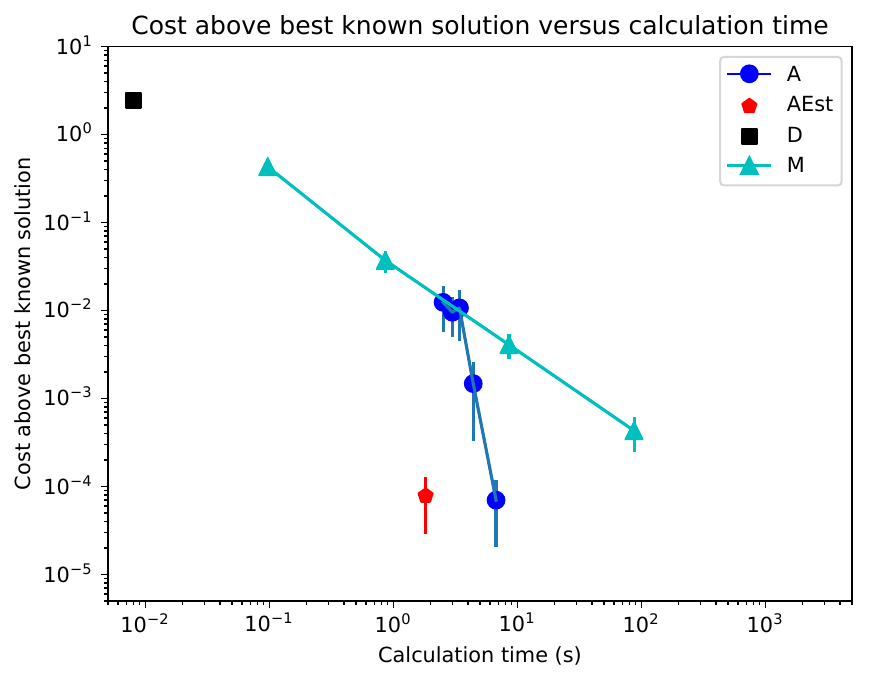}\label{sf:scen2_3}
	}
	%	\subfloat[$n=4$]{
	%		\includegraphics[width=0.4\textwidth]{scen2results_4}\label{sf:scen2_4}
	%	}
	
	%	\subfloat[$n=5$]{
	%		\includegraphics[width=0.4\textwidth]{scen2results_5}\label{sf:scen2_5}
	%	}
	\subfloat[6 uncontrolled robots per location ($n=6$)]{
		\includegraphics[width=0.44\textwidth]{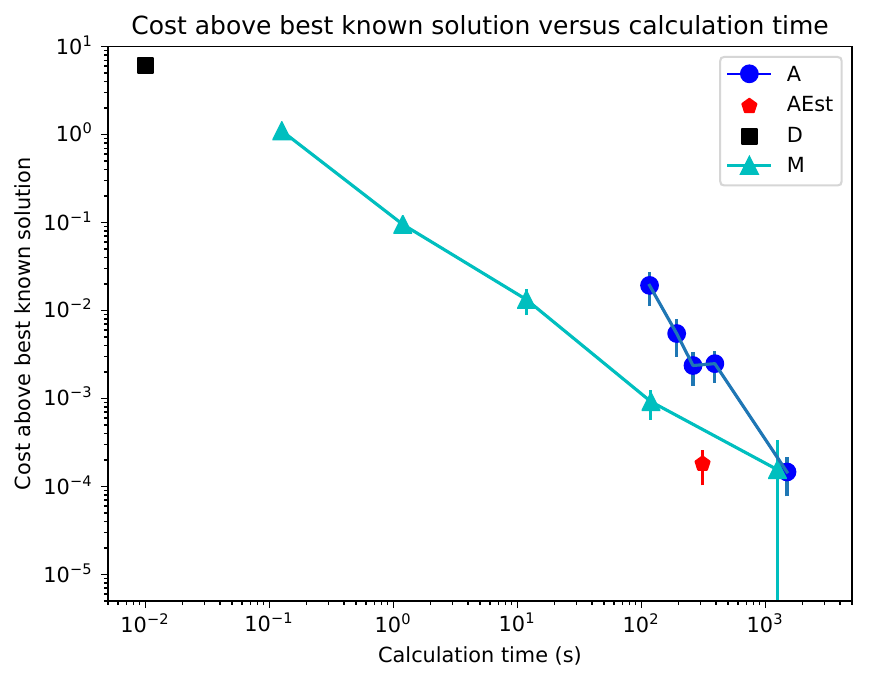}\label{sf:scen2_6}
	}
	
	\caption{Results for Scenario 2 showing the average of the cost for each scenario minus the best known cost versus the calculation time to find the optimal allocation. Values for the number of robots controlled by other entities at each location, $n$, of 1, 3, and 6, were used, as noted under each figure. The error bars show a 95\% confidence interval. The A method used thresholds on the sum of the probabilities of 80\%, 90\%, 95\%, 99\%, and 100\%. The MC method used 10, 100, 1,000, and 10,000 samples. In (c), 100,000 samples was also used. } \label{f:scen2_results}
\end{figure}

The number of orders of arrival considered by the A method suffers from factorial growth as the number of external robots is increased. These results suggest that, for cases where $n>6$ and the robots arrive at similar times, the M approach will outperform the A and AEst approaches. However, if it is possible to separate the robots into multiple groups where the probability that any of the robots in an earlier group will arrive after any of the robots in a later group is negligible, then the proposed approaches are potentially more suitable. The computationally expensive parts of the approaches were calculating the probability of an order of arrival (81\% of calculation time for A, 18\% of calculation time for AEst) and conditioning the arrival times on the order of arrival (14\% of calculation time for A, 58\% of calculation time for AEst). 

\section{Conclusion} \label{s:conc}

This paper presented an analytical framework for modelling timing uncertainty and resource contention in multi-robot scenarios. The framework was shown to significantly outperform deterministic approaches, and produce more accurate results than Monte Carlo methods with similar run-times. A key advantage of this approach over numerical methods is that it produces an accurate probability distribution of the result that can then be exploited in optimisation methods through approaches such as chance constrained programming. Certain aspects of the framework suffer from factorial computational complexity with the number of robots considered. Thus, the most promising avenues for future work focus on developing methods of segmenting the problem so that the number of robots considered by the framework at each resource is minimised. For example, splitting the robots into different arrival groups such that the arrival orders of each group can be considered separately. Other suggestions for future work include integrating this framework with an optimisation method such as branch and bound that considers partial solutions to further reduce the number of orders of arrival considered, and more accurately modelling other effects resulting from interactions between robot such as local collision avoidance. 

\appendices
\section{Proof of Theorem \ref{th:likely_arrival_order}} \label{ap:theorem_1_proof}

This section proves the claim of Theorem \ref{th:likely_arrival_order} that the most likely order of arrival in the case where the variances of the arrival times are equal is the order in which the mean arrival times are ascending. This can be restated as, for any number of independent normally distributed random variables with equal variances, the most likely sequence of the random variables resulting from sorting using their actual values is the sequence in which the random variables are sorted by their mean values. 

Consider two neighbouring random variables, $X$ and $Y$, in a sequence of random variables, where $X$ and $Y$ are independent and normally distributed, $X\sim \mathcal{N}(\mu_{X},\sigma^{2})$ and $Y\sim \mathcal{N}(\mu_{Y},\sigma^{2})$, with $\mu_{X} < \mu_{Y}$ and equal variances. This section will show that, regardless of the random variables before and after $X$ and $Y$, the probability that $X$ occurs before $Y$ is higher than the probability that $Y$ occurs before $X$. Let $a$ and $b$ be constants representing the values taken by the random variables immediately before and after $X$ and $Y$. Then this can be stated mathematically as:
\begin{multline} \label{eq:theorem1}
P(X < Y | a < X < b, a < Y < b) \\
> P(Y < X | a < X < b, a < Y < b). 
\end{multline}

\begin{lemma} \label{lemma:equality}
	$p_{X|a<X<b}(x) = p_{Y|a<Y<b}(x)$ for one and only one value of $x$. %if $\mu_{X} \ne \mu_{Y}$. 
\end{lemma}

\begin{proof}
	The probability distribution function $p_{X|a<X<b}(x)$ is a normal distribution that is truncated between $a$ and $b$. The PDF of a truncated normal distribution is given by:
	\begin{equation} \label{eq:truncated}
	p_{X|a<X<b}(x) = \frac{\frac{1}{\sigma} \phi \left( \frac{x-\mu_{X}}{\sigma}\right)}{\Phi\left( \frac{b-\mu_{X}}{\sigma}\right) - \Phi \left( \frac{a-\mu_{X}}{\sigma}\right)},
	\end{equation}
	for $a\le x\le b$, and 0 otherwise, where $\phi(.)$ and $\Phi(.)$ are defined in Eqs. (\ref{eq:normal_pdf}) and (\ref{eq:normal_cdf}) respectively. As the integrals of $p_{X|a<X<b}(x)$ and $p_{Y|a<Y<b}(x)$ are both 1, it follows that they must intersect at least once in the range $a<x<b$. The following proof shows that they intersect at most once:
	\begin{align}
	p_{X|a<X<b}(x) &= p_{Y|a<Y<b}(x) \\
	\frac{\frac{1}{\sigma} \phi \left( \frac{x-\mu_{X}}{\sigma}\right)}{\Phi\left( \frac{b-\mu_{X}}{\sigma}\right) - \Phi \left( \frac{a-\mu_{X}}{\sigma}\right)} & = \frac{\frac{1}{\sigma} \phi \left( \frac{x-\mu_{Y}}{\sigma}\right)}{\Phi\left( \frac{b-\mu_{Y}}{\sigma}\right) - \Phi \left( \frac{a-\mu_{Y}}{\sigma}\right)} \\
	\frac{ \phi \left( \frac{x-\mu_{X}}{\sigma}\right)}{\phi \left( \frac{x-\mu_{Y}}{\sigma}\right)} &= \frac{\Phi\left( \frac{b-\mu_{X}}{\sigma}\right) - \Phi \left( \frac{a-\mu_{X}}{\sigma}\right)}{\Phi\left( \frac{b-\mu_{Y}}{\sigma}\right) - \Phi \left( \frac{a-\mu_{Y}}{\sigma}\right)} \label{eq:lemma1_int}. 
	\end{align}
	
	Let:
	\begin{equation}
	c = \frac{\Phi\left( \frac{b-\mu_{X}}{\sigma}\right) - \Phi \left( \frac{a-\mu_{X}}{\sigma}\right)}{\Phi\left( \frac{b-\mu_{Y}}{\sigma}\right) - \Phi \left( \frac{a-\mu_{Y}}{\sigma}\right)}. 
	\end{equation}
	
	Substituting Eq. (\ref{eq:normal_pdf}) into Eq. (\ref{eq:lemma1_int}), taking the natural logarithm, and rearranging for $x$ gives:
	\begin{equation}
	x = \frac{2\sigma^{2}\ln(c) + \mu_{X}^{2} - \mu_{Y}^{2}}{2(\mu_{X} - \mu_{Y})}. 
	\end{equation}
	
\end{proof}

\begin{lemma} \label{lemma:fraction}
	$p_{X|a<X<b}(a) > p_{Y|a<Y<b}(a)$ and $p_{X|a<X<b}(b)<p_{Y|a<Y<b}(b)$.
\end{lemma}

\begin{proof}
	By Lemma \ref{lemma:equality}, $p_{X|a<X<b}$ and $p_{Y|a<Y<b}$ intersect only once. As $p_{X|a<X<b}$ and $p_{Y|a<Y<b}$ are probability distributions, they must be positive and have an integral of 1. Thus, it is sufficient to show that
	\begin{equation}
	\dfrac{p_{X|a<X<b}(b)}{p_{X|a<X<b}(a)} < \dfrac{p_{Y|a<Y<b}(b)}{p_{Y|a<Y<b}(a)},
	\end{equation}
	to prove the Lemma. Substituting in Eq. (\ref{eq:truncated}) gives:
	\begin{equation}
	\frac{ \phi \left( \frac{b-\mu_{X}}{\sigma}\right)}{\phi \left( \frac{a-\mu_{X}}{\sigma}\right)} < \frac{ \phi \left( \frac{b-\mu_{Y}}{\sigma}\right)}{\phi \left( \frac{a-\mu_{Y}}{\sigma}\right)}. 
	\end{equation}
	
	Substituting in Eq. (\ref{eq:normal_pdf}), taking the natural logarithm, and simplifying gives:
	\begin{equation}
	\mu_{X} < \mu_{Y}. 
	\end{equation}
	
\end{proof}

\begin{lemma} \label{lemma:expected}
	$E(X | a < X < b) < E(Y | a < Y < b)$. 
\end{lemma}

\begin{proof}
	Let $\lambda$ satisfy $p_{X|a<X<b}(\lambda) = p_{Y|a<Y<b}(\lambda)$. Then:
	\begin{align}
	E(X | a < X < b) &= \int_{a}^{b}x p_{X|a<X<b}(x) \textrm{d}x \\
	& = \int_{a-\lambda}^{b-\lambda} (x+\lambda) p_{X|a<X<b}(x+\lambda) \textrm{d}x \\
	& = \lambda + \int_{a-\lambda}^{b-\lambda} x p_{X|a<X<b}(x+\lambda) \textrm{d}x. 
	\end{align}
	
	By Lemmas \ref{lemma:equality} and \ref{lemma:fraction}: 
	\begin{equation}
	p_{X|a<X<b}(x+\lambda) > p_{Y|a<Y<b}(x+\lambda) \quad \forall x < 0, 
	\end{equation}
	\begin{equation}
	p_{X|a<X<b}(x+\lambda) < p_{Y|a<Y<b}(x+\lambda) \quad \forall x > 0. 
	\end{equation}
	
	Therefore:
	\begin{multline}
	\int_{a-\lambda}^{b-\lambda} x p_{X|a<X<b}(x+\lambda) \textrm{d}x \\
	< \int_{a-\lambda}^{b-\lambda} x p_{Y|a<Y<b}(x+\lambda) \textrm{d}x. 
	\end{multline}
	
\end{proof}

It is sufficient to show that $E(X | a < X < b) < E(Y | a < Y < b)$ to prove that $X | a < X < b$ is more likely to occur before $Y | a < Y < b$ than the opposite order. Thus, Lemma \ref{lemma:expected} proves Eq. (\ref{eq:theorem1}). Applying Eq. (\ref{eq:theorem1}) to all pairs of random variables leads to the conclusion that the most likely sequence will be the sequence resulting from sorting the random variables by their mean values, proving the theorem.

\bibliographystyle{IEEEtran}
\bibliography{ref}

\end{document}